\documentclass[11pt,a4paper]{article}

\usepackage[utf8]{inputenc}

\usepackage{fullpage} %more space for todo notes..

\usepackage{graphicx}

\usepackage{tabularx}

\usepackage{amsthm}
\usepackage{amsmath}
\usepackage{amssymb}
\usepackage{enumerate}
\usepackage{amsfonts}
\usepackage{mathtools}
\usepackage[sort,numbers,sectionbib]{natbib}

\usepackage{makecell}

\usepackage{tikz}
\usetikzlibrary{arrows,decorations.pathmorphing,decorations.pathreplacing,backgrounds,positioning,fit,matrix,arrows.meta,automata}
\usetikzlibrary{shapes,calc}

\tikzset{%
  apply style/.code={%
    \tikzset{#1}%
  }
}

\newcommand{\appsymb}{$\star$}
\newcommand{\appprooflink}[1]{{\hyperref[proof:#1]{\appsymb}}}

\usepackage{dsfont} % blackboard bold "1"

\usepackage{algorithm}
\usepackage[noend]{algpseudocode}
\algnewcommand\algorithmicinput{\textbf{Input:}}
\algnewcommand\algorithmicoutput{\textbf{Output:}}
\algnewcommand\Input{\item[\algorithmicinput]}
\algnewcommand\Output{\item[\algorithmicoutput]}

\usepackage{etoolbox}
\usepackage{tikz}
\usetikzlibrary{tikzmark}
\usetikzlibrary{calc}

\usepackage[menucolor=orange!40!black,filecolor=magenta!40!black,urlcolor=blue!40!black,linkcolor=red!40!black,citecolor=green!40!black,colorlinks,pagebackref]{hyperref}
\usepackage[capitalize,nameinlink]{cleveref}

\errorcontextlines\maxdimen

\newcommand{\ALGtikzmarkcolor}{black}% customise this, if you want
\newcommand{\ALGtikzmarkextraindent}{4pt}% customise this, if you want
\newcommand{\ALGtikzmarkverticaloffsetstart}{-.5ex}% customise this, if you want
\newcommand{\ALGtikzmarkverticaloffsetend}{-.5ex}% customise this, if you want
\makeatletter
\newcounter{ALG@tikzmark@tempcnta}

\newcommand\ALG@tikzmark@start{%
	\global\let\ALG@tikzmark@last\ALG@tikzmark@starttext%
	\expandafter\edef\csname ALG@tikzmark@\theALG@nested\endcsname{\theALG@tikzmark@tempcnta}%
	\tikzmark{ALG@tikzmark@start@\csname ALG@tikzmark@\theALG@nested\endcsname}%
	\addtocounter{ALG@tikzmark@tempcnta}{1}%
}

\def\ALG@tikzmark@starttext{start}
\newcommand\ALG@tikzmark@end{%
	\ifx\ALG@tikzmark@last\ALG@tikzmark@starttext
	\else
	\tikzmark{ALG@tikzmark@end@\csname ALG@tikzmark@\theALG@nested\endcsname}%
	\tikz[overlay,remember picture] \draw[\ALGtikzmarkcolor] let \p{S}=($(pic cs:ALG@tikzmark@start@\csname ALG@tikzmark@\theALG@nested\endcsname)+(\ALGtikzmarkextraindent,\ALGtikzmarkverticaloffsetstart)$), \p{E}=($(pic cs:ALG@tikzmark@end@\csname ALG@tikzmark@\theALG@nested\endcsname)+(\ALGtikzmarkextraindent,\ALGtikzmarkverticaloffsetend)$) in (\x{S},\y{S})--(\x{S},\y{E});%
	\fi
	\gdef\ALG@tikzmark@last{end}%
}

\apptocmd{\ALG@beginblock}{\ALG@tikzmark@start}{}{\errmessage{failed to patch}}
\pretocmd{\ALG@endblock}{\ALG@tikzmark@end}{}{\errmessage{failed to patch}}
\makeatother
\newtheorem{theorem}{Theorem}[section]
\newtheorem{lemma}[theorem]{Lemma}

\newtheorem{proposition}[theorem]{Proposition}

\theoremstyle{definition}

\crefname{observation}{Observation}{Observations}
\Crefname{observation}{Observation}{Observations}

\newcommand{\problemn}[1]{\textsc{#1}}
\newcommand{\problemdef}[3]{
	\begin{center}
	\begin{minipage}{0.95\textwidth}
		\noindent
		\problemn{#1}
		\vspace{5pt}\\
		\setlength{\tabcolsep}{3pt}
		\begin{tabularx}{\textwidth}{@{}lX@{}}
			\textbf{Input:}     & #2 \\
			\textbf{Question:}  & #3
		\end{tabularx}
	\end{minipage}
	\end{center}
}

\DeclarePairedDelimiterX{\abs}[1]{\lvert}{\rvert}{#1}

\newcommand{\mvert}{\;\middle|\;}

\newcommand{\NN}{\mathbb{N}} % natural numbers
\newcommand{\bigO}{\mathcal{O}} % Landau symbol

\newcommand{\TS}[1][$s,z$]{\textsc{Temporal (#1)-Separation}}

\newcommand{\fpt}{fixed-parameter tractable}

\newcommand{\tG}{\mathcal{G}} % temporal graph
\newcommand{\tE}{\mathcal{E}} % temporal edges
\newcommand{\isos}{I} % isolating vertices
\newcommand{\doms}{C} % dominating vertices
\newcommand{\var}{\Omega}

\newcommand{\Ss}{\mathcal{S}}

\newcommand{\true}{\texttt{true}}
\newcommand{\false}{\texttt{false}}

\usepackage{authblk}

\title{On Finding Separators in Temporal Split and Permutation~Graphs%
\thanks{
H.~Molter was supported by the DFG,
project MATE (NI 369/17), and by the ISF, grant No.~1070/20. Main part of this work was done while H.~Molter was affiliated with TU~Berlin.
M.~Renken was supported by the DFG,
project MATE (NI 369/17).}}

\author[1]{Nicolas Maack}
\author[2]{Hendrik Molter}
\author[1]{Rolf Niedermeier}
\author[1]{Malte Renken}

\date{ }

\affil[1]{TU Berlin, Algorithmics and Computational Complexity, Berlin, Germany,  
\texttt{nicolas.km.maack@campus.tu-berlin.de, \{rolf.niedermeier,~m.renken\}@tu-berlin.de}}

\affil[2]{Department of Industrial Engineering and Management, Ben-Gurion~University~of~the~Negev, 
Beer-Sheva, 
Israel, 
\texttt{molterh@post.bgu.ac.il}}

\begin{document}

\maketitle

\begin{abstract}
Removing all connections between two vertices $s$ and $z$ in a graph by removing a minimum number of vertices is a fundamental problem in algorithmic graph theory. This $(s,z)$-separation problem is well-known to be polynomial solvable and serves as an important primitive in many applications related to network connectivity.

We study the NP-hard \emph{temporal} $(s,z)$-separation problem on temporal graphs, which are 
graphs with fixed vertex sets but edge sets that change over discrete time
steps. We tackle this problem by restricting the layers (i.e., graphs characterized by edges that are present at a certain point in time) to specific graph classes.

We restrict the layers of the temporal graphs to be either all 
split graphs or all permutation graphs (both being perfect graph classes) and  provide both intractability and tractability results.
In particular, we show that in general
the problem remains NP-hard both on temporal split and temporal 
permutation graphs, but we also spot promising islands of fixed-parameter 
tractability particularly based on parameterizations that measure the amount of ``change over time''.

\smallskip

\end{abstract}

\section{Introduction}
Finding a smallest set of vertices whose deletion disconnects two 
designated vertices---the separation problem---is a fundamental problem in algorithmic graph 
theory. The problem, which is 
a backbone of numerous applications related to network connectivity,
is well-known to be polynomial-time solvable
in (static) graphs. Driven by the need of understanding and mastering
dynamically changing network structures, in recent years the 
study of temporal graphs---graphs with a fixed vertex set but edge
sets that may change over discrete time steps---has enjoyed an 
enormous growth. One  of the earliest systematic studies on 
temporal graphs dealt with the separation problem~\cite{KKK02}, where
it turned out to be NP-hard. 
This motivates the study of parameterized complexity aspects as well 
as of the complexity behavior on special temporal 
graph classes~\cite{Flu20,Zsc20}. 
Continuing and extending this line of research, we provide 
a first in-depth study on temporal versions 
of split and permutation graphs, two classes of perfect graphs
on which many generally NP-hard problems become  polynomial-time
solvable~\cite{BLS99,Gol04}. %Brandstaedt et al; Golumbic books
We present both intractability as well as (fixed-parameter) tractability 
results.

Formally, a \emph{temporal graph} is an ordered triple~$\tG = (V,\tE,\tau)$, where~$V$ denotes the set of vertices,
$\tE \subseteq \binom{V}{2} \times \{1,2,...,\tau\}$ the set of time-edges where $(\{v,w\},t) \in \tE$ represents an edge between vertices~$v$
and~$w$ available at time~$t$, and~$\tau \in \NN$ is the maximum time label. We can think of it as a series of $\tau$~static
graphs, called \emph{layers}. 
The graph containing the union of the edges of all layers is called 
the \emph{underlying graph} of~$\tG$.

Recently, connectivity and path-related problems 
have been extensively studied on temporal
graphs~\cite{ComputeTemporalPaths,Zsc20,Flu20,MMS19,BentertHNN20,BMNR20,Erlebach0K15,enright2021assigning,EnrightMMZ19}.
In the temporal setting, paths, walks, and reachability are 
defined in a time-respecting way~\cite{KKK02}: 
		A 
		\emph{temporal ($s,z$)-walk} (or \emph{temporal walk})  
		of length~$k$ from vertex $s=v_0$ to vertex $z=v_k$ in a temporal graph~$\tG$ is a sequence
$P = \left(\left(v_{i-1},v_i,t_i\right)\right)_{i=1}^k$
such that $(\{v_{i-1}, v_i\}, t_i)$ is a time-edge of~$\tG$ for all $i \in \{1,2,\dots,k\}$
and $t_i \leq t_{i+1}$ for all $i\in \{1,2,\ldots,k-1\}$.%
\footnote{Such walks are also called ``non-strict'', whereas ``strict'' walks require~$t_i < t_{i+1}$. We focus on non-strict walks in this work.}
A temporal walk is a temporal \emph{path} if it
visits every vertex at most once.

We study the \TS{} problem. 
Here, a \emph{temporal ($s,z$)-separator} is a set of vertices (not containing $s$ and $z$) whose removal destroys all temporal paths from~$s$ to~$z$.

\problemdef{Temporal ($s,z$)-Separation}{
	A temporal graph~$\tG = (V,\tE,\tau)$, two distinct vertices~$s,z \in V$, and~$k \in\NN$.
}{
	Does~$\tG$ admit a temporal ($s,z$)-separator of size at most~$k$?
}\label{pr:ts}

\TS{} is 
NP-complete~\cite{KKK02} and W[1]-hard when parameterized by the separator size
$k$~\cite{Zsc20}.
On the positive side, one can verify a solution in $\bigO(\abs{\tG})$ time (see e.g.\ \citet{ComputeTemporalPaths}).
\citet{Zsc20} investigated the differences between the
computational complexity of finding temporal separators that remove (non-strict)
temporal paths vs.\ \emph{strict} temporal paths. 
\citet{Flu20} studied the impact
of restrictions on the layers or the underlying graph on the computational complexity of \TS{} and found that even under severe
restrictions \TS{} remains NP-complete.
In particular, the problems stays NP-complete and W[1]-hard when
parameterized by the separator size even if every layer contains only one
edge and for several restrictions of the underlying graph~\cite{Flu20}.
They further investigated the case where every layer of the temporal
graph is a unit interval graph and obtained fixed-parameter tractability for
\TS{} when parameterized by the lifetime~$\tau$ and the so-called ``shuffle
number'', a parameter that measures how much the relative order of the intervals changes over time. This result initiated research on the amount of ``change over time'' measured by a parameter that is tailored to the graph class into which all layers fall.

In our work, we follow this paradigm of restricting each layer to a certain graph class
and measuring the amount of change over time with parameters that are tailored to the graph class of the layers. 
More specifically, we investigate the complexity 
of \TS{} on temporal graphs where
every layer is a \emph{split graph} or every layer is a 
\emph{permutation graph}. 

\smallskip

\noindent\textbf{Temporal split graphs.}
In a split graph, the vertex set can be partitioned
into a clique and an independent set. 
Split graphs can be used to model an idealized form of core-periphery structures, in which there exists a densely connected core and a periphery that only has connections
to that core \cite{BE00}. 
Core-periphery structures can be observed in social contact networks in which one group of people meets
at some location, forming a fully connected core. Meanwhile, other people associated with the group may
have contact with some of its members, but otherwise do not have any interactions relevant to the observed network.
Such applications are naturally subject to change over time, for example due to vertices entering or leaving the network.%
\footnote{While the vertex set of a temporal graph formally remains unchanged, isolated vertices are equivalent to non-existing vertices as far as separators are concerned.}
This temporal aspect is captured by \emph{temporal split graphs},
where each layer has a separate core-periphery split.

We prove that \TS{} remains NP-complete on temporal split graphs, even with only four layers.
On the positive side, we show that \TS{} is solvable in polynomial time on
temporal graphs where %every layer is a split graph and 
the partition of the vertices into a clique and an
independent set stays the same in every layer. We use this as a basis for a
``distance-to-triviality''-parameterization~\cite{GHN04,Nie06}, also motivated by the assumption that in application cases the core-periphery structure of a network will roughly stay intact over time with only few changes.
Intuitively, we parameterize on how
many vertices may switch between the two parts over time 
and use this parameterization to obtain
fixed-parameter tractability results. Formally, we show
fixed-parameter tractability for the combined parameter ``number of vertices
different from $s$ and~$z$ that switch between the 
clique and the independent set at
some point in time'' and the lifetime~$\tau$.

\smallskip

\noindent\textbf{Temporal permutation graphs.}
A permutation graph on an ordered set of vertices is
defined by a permutation of that ordering. Two
vertices are connected by an edge if their relative order is inverted by the permutation.
They were introduced by \citet{EvenPermutation72}
and appear in integrated circuit design \cite{SenDG92},
memory layout optimization \cite{EvenPermutation72},
and other applications \cite{Gol04}.
In a \emph{temporal permutation graph}, the edges of each layer are given by a separate permutation.
We prove that \TS{} remains
NP-complete on temporal permutation graphs.
We then parameterize on how much the permutation 
changes over time to obtain fixed-parameter
tractability results.
We use the \emph{Kendall tau} distance~\cite{Ken38} to measure
the dissimilarity of the permutations. The Kendall tau distance is a metric
that counts the number of pairwise disagreements between two total orderings; it is also known as ``bubble sort distance''. 
More precisely, we obtain
fixed-parameter tractability for the combined parameter ``sum of Kendall tau distances between consecutive permutations'' and the separator size~$k$. We remark that in a similar context, the Kendall tau distance has also been used by \citet{Flu20} to measure the amount of change over time in a temporal graph.

%%%%%%%%%%%%%%%%%%%%%%%%%%%%%%SPLIT GRAPHS%%%%%%%%%%%%%%%%%%%%%%%%%%%%%%%%%%

\section{Split Graphs}\label{chap:split}

In this section, we study the computational complexity of finding temporal
separators in \emph{temporal split graphs}.
Split graphs represent an idealized model of core-periphery structures with a well-connected core and a periphery only connected to that core \cite{BE00}.
They also constitute the majority of all chordal graphs~\cite{bender1985}.

Formally, a graph~$G=(V,E)$ is called a \emph{split graph} if $V$~can be partitioned into two sets~$\doms, \isos$
such that~$\doms$ induces a clique and $\isos$ induces an independent set.
Then, $(\doms, \isos)$~is called a \emph{split partition} of~$G$.
In general, a split graph may admit multiple split partitions.
A \emph{temporal split graph} is a temporal graph~$\tG$ of which every layer is a split graph.
A \emph{temporal split partition} $(\doms_t, \isos_t)_{t=1}^\tau$ then contains a split partition of every layer of $\tG$.

\subsection{Hardness Results}\label{chap:split:hardness}

The fact that \TS{} on temporal split graphs is NP-hard can be derived from a result of \citet{Flu20}
stating that \TS{} is hard on temporal graphs containing a single edge per layer.
This is due to the fact that a graph with a single edge is clearly a split graph.

We now strengthen this result, showing that \TS{} is NP-hard on temporal split
graphs with only a constant number of layers. We do this by building on 
a reduction
by \citet{Zsc20} showing NP-hardness of \TS{} on general temporal graphs.

Here and in the following we use ``separator'' as a shorthand for ``temporal $(s,z)$-separator'' when no ambiguity arises.

\begin{theorem}%[\appprooflink{thm:ts_tau_np}]
\label{thm:ts_tau_np}
\TS{} is NP-hard on temporal split graphs with four layers.
\end{theorem}
%\appendixproof{thm:ts_tau_np}{
\begin{proof}
We prove this by building on the proof by \citet{Zsc20} showing NP-hardness of \TS{} on general temporal graphs.
Their proof employs a polynomial-time reduction from \textsc{Vertex Cover}. We extend this reduction so that it produces a \TS{} instance
on a temporal split graph with~$\tau=4$.

Let~$\mathcal{I}:=(G=(V,E),k)$ be a \textsc{Vertex Cover} instance and~$n:=\abs{V}$.
Construct a \TS{} instance~$\mathcal{I}':=(\tG'=(V',\tE',2),s,z,n+k)$, where~$V'=V\cup\{s_v,z_v\mid v\in V\}\cup\{s,z\}$
are the vertices and the time-edges are defined as
\begin{align*}
\tE':={}&\{(\{s,s_v\},1),(\{s_v,v\},1),(\{v,z_v\},2),(\{z_v,z\},2),(\{s,v\},2),(\{v,z\},1)\mid v\in V\}\cup\\
&\{(\{s_v,z_w\},1),(\{s_w,z_v\},1)\mid\{v,w\}\in E\}.
\end{align*}
Note that \citet{Zsc20} showed that deciding $\mathcal{I}'$ is NP-complete.

The idea is now to divide each of the layers of $\tG'$ into two layers in a way that turns all layers into split graphs
while not changing the set of temporal $(s,z)$-separators.
More specifically, we construct another \TS{} instance~$\mathcal{I}'':=(\tG''=(V',\tE'',4),s,z,n+k)$, where the time-edges are defined as
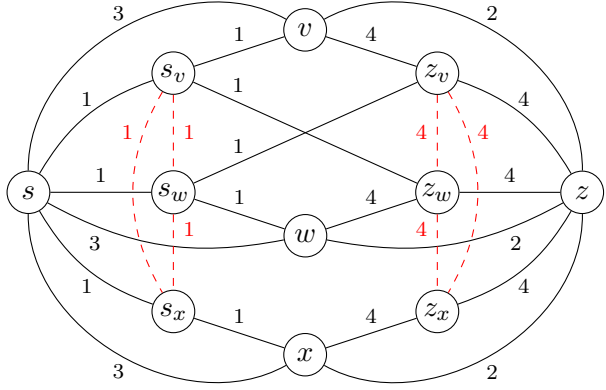
\begin{figure}[t]
	\centering
	\tikzstyle{alter}=[circle, minimum size=16pt, draw, inner sep=1pt] 
	\tikzstyle{clique}=[red, dashed]
	\begin{tikzpicture}
		\node[alter] at (0,0) (s) {$s$};
		\node[alter, right = 8ex of s] (sw) {$s_w$};
		\node[alter, below = of sw] (sx) {$s_x$};
		\node[alter, above = of sw] (sv) {$s_v$};
		\node[alter, below right = 1ex and 8ex of sw] (w) {$w$};
		\node[alter, below right = 1ex and 8ex of sx] (x) {$x$};
		\node[alter, above right = 1ex and 8ex of sv] (v) {$v$};
		\node[alter, above right = 1ex and 8ex of w] (zw) {$z_w$};
		\node[alter, below = of zw] (zx) {$z_x$};
		\node[alter, above = of zw] (zv) {$z_v$};
		\node[alter, right = 8ex of zw] (z) {$z$};
		
		\draw (s) edge[bend left=20] node[midway, anchor=south] {$\scriptstyle 1$} (sv);
		\draw (s) edge node[midway, anchor=south] {$\scriptstyle 1$} (sw);
		\draw (s) edge[bend right=20] node[midway, anchor=north] {$\scriptstyle 1$} (sx);
		\draw (s) edge[bend left=60] node[midway, anchor=south] {$\scriptstyle 3$} (v);
		\draw (s) edge[bend right=20] node[pos=0.2, anchor=north] {$\scriptstyle 3$} (w);
		\draw (s) edge[bend right=60] node[midway, anchor=north] {$\scriptstyle 3$} (x);
		\draw (sv) edge node[midway, anchor=south] {$\scriptstyle 1$} (v);
		\draw (sv) edge node[pos=0.2, anchor=south] {$\scriptstyle 1$} (zw);
		\draw (sw) edge node[midway, anchor=south] {$\scriptstyle 1$} (w);
		\draw (sw) edge node[pos=0.2, anchor=south] {$\scriptstyle 1$} (zv);
		\draw (sx) edge node[midway, anchor=south] {$\scriptstyle 1$} (x);
		
		\draw (z) edge[bend right=20] node[midway, anchor=south] {$\scriptstyle 4$} (zv);
		\draw (z) edge node[midway, anchor=south] {$\scriptstyle 4$} (zw);
		\draw (z) edge[bend left=20] node[midway, anchor=north] {$\scriptstyle 4$} (zx);
		\draw (z) edge[bend right=60] node[midway, anchor=south] {$\scriptstyle 2$} (v);
		\draw (z) edge[bend left=20] node[pos=0.2, anchor=north] {$\scriptstyle 2$} (w);
		\draw (z) edge[bend left=60] node[midway, anchor=north] {$\scriptstyle 2$} (x);
		\draw (zv) edge node[midway, anchor=south] {$\scriptstyle 4$} (v);
		\draw (zw) edge node[midway, anchor=south] {$\scriptstyle 4$} (w);
		\draw (zx) edge node[midway, anchor=south] {$\scriptstyle 4$} (x);
		
		\draw[clique] (sv) edge node[midway, anchor=west] {$\scriptstyle 1$} (sw);
		\draw[clique] (sw) edge node[pos=0.2, anchor=west] {$\scriptstyle 1$} (sx);
		\draw[clique] (sv) edge[bend right=30] node[pos=0.2, anchor=east] {$\scriptstyle 1$} (sx);
		
		\draw[clique] (zv) edge node[midway, anchor=east] {$\scriptstyle 4$} (zw);
		\draw[clique] (zw) edge node[pos=0.2, anchor=east] {$\scriptstyle 4$} (zx);
		\draw[clique] (zv) edge[bend left=30] node[pos=0.2, anchor=west] {$\scriptstyle 4$} (zx);
	\end{tikzpicture}
	\caption{The graph~$\tG''$ for~$G=(\{v,w,x\},\{\{v,w\}\})$ in \cref{thm:ts_tau_np}.
	The edges from~$\tE_\alpha$ are in solid black, and the edges from~$\tE_C$ are in red and dashed.}
	\label{fig:ts_tau_np}
\end{figure}
\begin{align*}
\tE'':={}&\tE_\alpha\cup\tE_C, \text{where}\\
\tE_\alpha:={}&\{(\{s,s_v\},1),(\{s_v,v\},1),(\{v,z_v\},4),(\{z_v,z\},4),(\{s,v\},3),(\{v,z\},2)\mid v\in V\}\cup\\
&\{(\{s_v,z_w\},1),(\{s_w,z_v\},1)\mid\{v,w\}\in E\},\\
\tE_C:={}&\{(\{s_v,s_w\},1),(\{z_v,z_w\},4)\mid v,w\in V,v\neq w\}.
\end{align*}
The time-edges in~$\tE_\alpha$ are the same edges as in~$\tE'$, except that they are distributed over four layers.
The time-edges in~$\tE_C$ induce cliques in layers~1 and~4, turning them into split graphs.
The constructed temporal graph is visualized in \cref{fig:ts_tau_np}.

Note that~$\abs{\tE''}=\abs{\tE'}+2{n \choose 2}$ and~$\mathcal{I}''$ can be computed in polynomial time.
All four layers of~$\tG''$ are split graphs: the time-edges in~$\tE_C$ exactly form a clique each in layers~1 and~4;
in layer~2, $\{z\}$ is the clique, and in layer~3, $\{s\}$ is the clique.
The time-edges in~$\tE_\alpha$ all have a vertex from their layer's clique as one endpoint, thus not connecting
any two vertices outside of it.

\paragraph*{Correctness.}
We prove that~$\mathcal{I}'$ is a yes-instance if and only if~$\mathcal{I}''$ is a yes-instance
by showing that a vertex set is a ($s,z$)-separator of~$\tG'$ if and only if it is also
a ($s,z$)-separator of~$\tG''$.

First, we note that the temporal ($s,z$)-paths in~$\tE'$ and~$\tE_\alpha$ are identical except for their time labels:
every time-edge in layer~1 of~$\tG'$ has a corresponding time-edge from~$\tE_\alpha$ between the same vertices in layer~1 or~2 of~$\tG''$,
and in layer~2 of~$\tG'$, each has a corresponding time-edge from~$\tE_\alpha$ in layer~3 or~4 of~$\tG''$.
All time-edges in~$\tE_\alpha$ correspond to one in~$\tE'$ in this way.
This means that every temporal path in~$\tE_\alpha$ has a corresponding path in~$\tE'$
that also has ascending time labels. Conversely, a temporal path in~$\tE'$
will have a corresponding path in~$\tE_\alpha$ that first only passes through time-edges with labels~1 or~2, and then through
time-edges with labels~3 or~4. Also, all time-edges with label~2 in~$\tE_\alpha$ are connected to~$z$,
and all time-edges with label~3 in~$\tE_\alpha$ are connected to~$s$. So an ($s,z$)-path cannot pass any edges after
an edge in layer~2 or before an edge in layer~3, making impossible any ($s,z$)-path with an edge from layer~2 before one from layer~1,
or an edge from layer~4 before one from layer~3. Hence, all paths in~$\tE_\alpha$ corresponding to a temporal ($s,z$)-path
in~$\tE'$ are also valid temporal ($s,z$)-paths.

$\Rightarrow$: We show that the addition of~$\tE_C$ does not result in any new temporal ($s,z$)-paths that can not be shortened
to one already existing in~$\tE_\alpha$, and by extension, in~$\tE'$. In layer~1, there is a time-edge from~$s$ to every vertex
that is part of the clique. Any path that takes one or more of the time-edges in layer~1 of~$\tE_C$ can therefore
be shortened to a path that goes directly from~$s$ to the last vertex reached by taking one of the new time-edges, eliminating all
of these new edges from the path. Equivalently, all vertices in the clique in layer~4 have a time-edge to~$z$. So instead
of taking any of the time-edges leading to other members of the clique, any ($s,z$)-path could just directly take the edge to~$z$.
This means that there are no temporal ($s,z$)-paths in~$\tG''$ that do not get cut by an ($s,z$)-separator of~$\tG'$.

$\Leftarrow$: Any ($s,z$)-separator of~$\tG''$ must cut all temporal ($s,z$)-paths in~$\tE_\alpha$,
thereby also cutting all temporal ($s,z$)-paths in~$\tE'$, making it a
($s,z$)-separator of~$\tG'$, too.\qed
\end{proof}
%}%appendixproof

Theorem~\ref{thm:ts_tau_np} shows that \TS{} on temporal split graphs is NP-hard for~$\tau\geq 4$.
Clearly, \TS{} is polynomial-time solvable for~$\tau=1$.
The computational complexity for the cases $\tau=2$ and $\tau=3$ remains open.

\subsection{Fixed-Parameter Tractability Results}\label{chap:split:switch}

The defining characteristic of temporal split graphs is that for each layer~$t$ they can be split into a clique~$\doms_t$ and an independent set~$\isos_t$.
While temporal split graphs allow for changes of this partition,
for several application scenarios described in the introduction it 
is reasonable to assume that only a few of these changes occur
while most vertices retain their role throughout all layers.

We will prove that \TS{} can be solved efficiently in this setting, that is,
if the number of \emph{switching vertices} $\bigcup_{t, t'} \doms_t \cap \isos_{t'}$ is low.
To this end, we first prove that \TS{} is polynomial-time solvable if the only switching vertices are~$s$~and~$z$.

\begin{lemma}\label{thm:noway}
Let $\tG=(V\cup\{s,z\},\tE,\tau)$ be a temporal split graph with a given temporal split partition having
no switching vertices except possibly $s$~and~$z$.
Then all minimal temporal ($s,z$)-separators in~$\tG$ can be found in $\bigO(\abs{\tG} \cdot\tau)$~time.
\end{lemma}

\begin{proof}
We assume that there is never an edge between $s$~and~$z$, otherwise the problem is trivial.
Let the given temporal split partition be $(\doms_t, \isos_t)_t$.
Let $\doms := \doms_t \setminus \{s, z\}$ and $\isos := \isos_t \setminus \{s, z\}$ (for some and thus all layers~$t$).
We show that all minimal temporal $(s, z)$-separators are given by the set
\begin{align*}
	\Ss:={}&\left\{\bigcup_{0<t\le i}(N_{G_t}(s)\cap\doms) \cup\bigcup_{i<t\le\tau}(N_{G_t}(z)\cap\doms)\cup T
	\mvert  0\le i\le\tau\right\},\text{ where}\\
	T:={}& \bigcup_{1\le t\le t'\le\tau} N_{G_t}(s) \cap N_{G_{t'}}(z) .
\end{align*}
The set~$\Ss$ can be constructed in~$\bigO(\abs{\tG}\cdot\tau)$ time
and contains at most~$\tau+1$ elements.
This proves the stated time bound.
It remains to verify that $\Ss$ contains all minimal separators.

First, note that $T$~contains exactly those vertices which form temporal ($s,z$)-paths of length~2.
Thus $T$ has to be contained in any separator.

So it only remains to consider temporal $(s, z)$-paths of length at least~$3$.
If such a path~$P$ contains a vertex $v \in I$,
then let $(\{u,v\}, t)$ and  $(\{v,w\}, t')$ be the two time-edges of $P$ containing~$v$.
Then we must have $u \in C$ or $w \in C$ and thus the above two time-edges can be replaced 
by either $(\{u, w\}, t)$ or $(\{u, w\}, t')$, shortening~$P$ by one.
Consequently, if $\tG$ contains a temporal $(s,z)$-path of length at least 3,
then there is also a temporal $(s,z)$-path in~$\tG - I$.\footnote{We denote by $\tG - X$ the temporal graph resulting
from removing vertices in $X$ from the vertex set of $\tG$.}
So it suffices to consider temporal paths in~$\tG - I$.

Thus it becomes clear that each element of~$\Ss$ is in fact a temporal $(s,z)$-separator.
Now let~$S$ be an arbitrary temporal $(s,z)$-separator and $i$ maximal with $S \supseteq \bigcup_{0<t\le i}N_{G_t}(s)\cap\doms$.
Thus,~$S$~does not contain some vertex~$v\in N_{G_{i+1}}(s) \cap \doms$.
Then, $\tG{} - S$ contains a path from $s$ to every vertex in $\doms \setminus S$ since~$v$ is connected to all other vertices in~$\doms$.
So starting from layer~$i+1$, all edges to~$z$ from a vertex in~$\doms \setminus S$ would complete a temporal path from~$s$ to~$z$.
Hence, $S \supseteq \bigcup_{i<t\le\tau}N_{G_t}(z)\cap\doms$,
concluding the proof.\qed
\end{proof}

\newcommand{\splits}{\Psi}
\newcommand{\suredoms}{\hat{\doms}}
\newcommand{\sureisos}{\hat{\isos}}
Next, we 
first show that this is still an NP-hard problem and then give an FPT-algorithm for the solution size parameter. 
\begin{proposition}%[\appprooflink{thm:switch-num-hard}]
\label{thm:switch-num-hard}
For temporal split graphs, it is NP-hard to compute a minimum-size set of 
switching vertices. 
\end{proposition}
%\appendixproof{thm:switch-num-hard}{
\begin{proof}
We present a polynomial-time reduction from \textsc{Vertex Cover}~\cite{GJ79}.
Let $(H=(U,F), h)$ be an instance of \textsc{Vertex Cover}, then we construct a temporal split graph $\tG=(V,\tE,\tau)$ as follows.
We set $V:=\{u_1, u_2\mid u\in U\}$.
Now we create $|F|+1$ layers.
In the first layer, we create a clique on the vertices~$\{u_1 \mid u\in U\}$ and connect $u_1$ with $u_2$ for all $u\in U$. Formally,
\[
E_1(\tG):=\{\{u_1,v_1\}\mid u,v\in U\land u\neq v\} \cup \{\{u_1,u_2\}\mid u\in U\}.
\]
Now let the edges of $H$ be ordered in some fixed but arbitrary way, that is, $F=\{e_1, e_2, \ldots, e_{|F|}\}$. For $i\in [|F|]$, we create layer $i+1$ with the following edge set:
\[
E_{i+1}(\tG):=\{\{u_1,v_1\}\mid u,v\in U\land u\neq v\land \{u,v\}\neq e_i\}.
\]
This graph can clearly be constructed in polynomial time and it is easy to check that it is in fact a temporal split graph. 

\paragraph*{Correctness.} In the following we show that the cardinality of a minimum-size set of switching vertices for $\tG$ is at most $h$ if and only if $H$ admits a vertex cover of size at most $h$.

We start with some observations about the constructed temporal split graph~$\tG$. In the first layer, we have a unique valid partition of the vertices into a clique and an independent set: $C_1=\{u_1 \mid u\in U\}$ forms the clique and $I_1=\{u_2\mid u\in U\}$ forms the independent set. Vertices in $I_1$ may stay in the independent set in all other layers as well, however in every layer exactly one pair of vertices in~$C_1$ misses the connecting edge, hence at least one of them needs to be in the independent set of that layer. Intuitively, these vertices which need to switch correspond to a vertex cover in $H$. 

$\Rightarrow$: Assume $H$ contains a vertex cover $X\subset U$ of size at most $h$. Then we claim that the set $Y=\{u_1\mid u\in X\}$ forms a set of switching vertices of size $h$ for $\tG$. The following partitions of the vertices into a clique and an independent set for all layer witness this. As already described earlier, we have that $C_1=\{u_1 \mid u\in U\}$ and $I_1=\{u_2\mid u\in U\}$. For layer $i$ with $i>1$ we set $C_i=C_1\setminus\{u_1\mid u\in X \cap e_{i-1}\}$ and $I_i=I_1\cup\{u_1\mid u\in X \cap e_{i-1}\}$. First, we can observe that the vertices in~$Y$ are the only switching vertices. Furthermore, since $X$ is a vertex cover in $H$ we have that for each layer $i$, at least one of the two vertices from $C_1$ that are not connected by an edge is in $I_i$. Otherwise the edge $e_{i-1}$ would not be covered by~$X$.

$\Leftarrow$: Assume we find a set of switching vertices $Y$ of size at most $h$ for $\tG$. First, observe that $Y\cap I_1=\emptyset$ since in none of the layers, a vertex from $I_1$ can be in the clique. We claim that $X = \{u \mid u_1\in Y\}$ is a vertex cover for $H$. Assume for contradiction that there is an edge $e_i\in F$ such that $e_i\cap X=\emptyset$. Let $e_i=\{u,v\}$. Then we have that in layer $i+1$ vertices $u_1, v_1$ are both in the clique, since they are in $C_1$ and not switching. However, by construction there is no edge between $u_1$ and $v_1$ in layer $i+1$. Hence, not both vertices can be part of the clique in this layer. A contradiction.\qed
\end{proof}
%}%appendixproof

\begin{proposition}%[\appprooflink{thm:compute-switching}]
\label{thm:compute-switching}
For a temporal split graph~$\tG=(V,\tE,\tau)$,
one can find a temporal split partition minimizing the number of switching vertices in 
$\bigO(\abs{\tE} + \tau\cdot\abs{V} + \abs{V}^2 \cdot (1.2738^p + p \cdot
\abs{V}))$~time, where $p$ is the minimum number of switching vertices.
\end{proposition}
%\appendixproof{thm:compute-switching}{
\begin{proof}
Let $G_t$ be any layer of $\tG$.
By a result of \citet{Heggernes06}, we can find in $\bigO(\abs{G_t})$~time a partition
$V = \suredoms_t \cup \sureisos_t \cup Q_t$ (each part possibly empty)
such that $\suredoms_t \subseteq \doms$ and $\sureisos_t \subseteq \isos$ for every split partition $(\doms, \isos)$ of~$G_t$.
Furthermore,~$Q_t$ forms either a clique or an independent set and every vertex of~$Q_t$ is adjacent to all vertices of $\suredoms_t$ and to no vertex of $\sureisos_t$.
If $Q_t$ is a clique, then all possible split partitions of~$G_t$ are clearly given by
\[
	\{(V \setminus \isos, \isos) \mid \sureisos_t \subseteq \isos \subseteq \sureisos_t \cup Q_t, \abs{\isos \cap Q_t} \leq 1\}
\]
and if~$Q_t$ is an independent set, then all split partitions of~$G_t$ are given by
\[
	\{(\doms, V \setminus \doms) \mid \suredoms_t \subseteq \doms \subseteq \suredoms_t \cup Q_t, \abs{\doms \cap Q_t} \leq 1 \}.
\]

Suppose now we are given a temporal split partition~$(\doms_t, \isos_t)$ of~$\tG$.
Denote by~$T$ the set of switching vertices of this partition,
and by $\doms := \bigcap_t \doms_t$ and $\isos := \bigcap_t \isos_t$ the set of non-switching vertices.

Let $Q_\cap := \bigcap_{t=1}^\tau Q_t$.
If there exist $t, t'$ such that $Q_\cap$ is a clique in $G_t$ but an independent set in $G_{t'}$,
then all but two vertices in~$Q_\cap$ must be switching vertices,
as $Q_\cap \cap I_t$ and $Q_\cap \cap C_{t'}$ can each contain at most one vertex.

Clearly, all vertices in $S := \bigcup_{t \neq t'} \suredoms_t \cap \sureisos_{t'}$ must also be switching vertices.
The remaining vertices $\tilde{V} := V \setminus (Q_\cap \cup S)$ can be partitioned into two sets.
The first one,~$\tilde{\doms} := \bigcap_t (\suredoms_t \cup Q_t) \setminus Q_\cap$
contains vertices which are sometimes in~$\suredoms_t$ and sometimes in~$Q_t$.
The second one,~$\tilde{\isos} := \bigcap_t (\sureisos_t \cup Q_t) \setminus Q_\cap$
contains vertices which are sometimes in~$\sureisos_t$ and sometimes in~$Q_t$.

We build an auxiliary graph $A_\doms$ as follows:
Start with the vertex set $\tilde{\doms}$.
Additionally, if there is a vertex $q_\doms \in Q_\cap \cap \doms$,
then also add $q_\doms$ to the vertices of $A_\doms$.
Then connect any two vertices $v, w$ in $A_\doms$ if and only if they are \emph{not} connected in some layer of $\tG$.
Analogously, we build another graph $A_\isos$.
The vertex set of $A_\isos$ is given by $\tilde{\isos}$ and, possibly,
the unique vertex $q_\isos \in Q_\cap \cap \isos$.
We connect two vertices $v, w$ in $A_\isos$ if and only if they are connected in some layer of $\tG$.

Observe that $T \cap V(A_\doms)$ forms a vertex cover of $A_\doms$:
if $\{v, w\} \in E(A_\doms)$, then there is a layer with $\{v, w\} \notin E(G_t)$,
thus at least one of $v, w$ must be in $\isos_t$ and thus in $T$,
since every vertex of $A_\doms$ is contained in some $C_{t'}$.
Analogously, $T \cap V(A_\isos)$ forms a vertex cover of $A_\isos$:
if $\{v, w\} \in E(A_\isos)$, then there is a layer with $\{v, w\} \in E(G_t)$,
therefore one of $v, w$ is in $\doms_t \cap \tilde{\isos} \subseteq T$.
Obviously, these vertex covers do not contain $q_\doms$ or $q_\isos$.

We claim that we can use the converse direction to obtain a temporal split partition of $\tG$.
First, we compute $\suredoms_t$, $\sureisos_t$, and $Q_t$ for each layer~$G_t$.
Then we select up to two vertices $q_\doms, q_\isos \in Q_\cap$.
Having these, we build the two auxiliary graphs~$A_\doms$ and~$A_\isos$ as
above.
We then compute two vertex covers~$T_\doms$ and~$T_\isos$, excluding $q_\doms$ and $q_\isos$ if they exist.
The temporal split partition $(\doms_t, \isos_t)_t$ is then constructed as follows.
Clearly, $\doms_t \supseteq \suredoms_t$ and $\isos_t \supseteq \sureisos_t$.
If $Q_t = \emptyset$, then this already defines~$(C_t, I_t)$.

If $Q_t$ is a clique in $G_t$, then we add to $\isos_t$
any vertices in $V(A_\isos) \setminus T_\isos$
and all remaining vertices to $\doms_t$.
Since $T_\isos$ is a vertex cover of $A_\isos$, $Q_t \cap (V(A_\isos) \setminus T_\isos)$ contains at most one vertex, 
making this a valid split partition of~$G_t$.%\todo{RN: This makes no sense.}
If $Q_t$ is an independent set in $G_t$, then we analogously add to $\doms_t$
any vertices in $V(A_\doms) \setminus T_\doms$
and all remaining vertices to $\isos_t$.

It is not difficult to check that the switching vertices of the resulting temporal split partition are exactly the vertices in $S \cup T_\doms \cup T_\isos \cup (Q_\cap \setminus \{q_\doms, q_\isos\})$.
To ensure that the resulting temporal split partition has a minimum number of switching vertices, we can simply try all possible choices for $q_\doms$ and $q_\isos$,
each time computing minimum vertex covers $T_\isos$ and $T_\doms$.

The required time to compute all $\suredoms_t, \sureisos_t, Q_t$
as well as $Q_\cap$, $S$, $A_\doms$ and $A_\isos$ is in $\bigO(\sum_t \abs{G_t}) \subseteq \bigO(\abs{\tE} + \tau\abs{V})$.
For each of the $\bigO(\abs{V}^2)$ possible ways of choosing $q_\doms, q_\isos$,
we need to compute minimum size vertex covers for $A_\doms$ and $A_\isos$.
Each of these computations is possible in $\bigO(1.2738^p + p\cdot\abs{V})$~time \cite{ChenVertexCover}.
The overall time requirement is thus in $\bigO(\abs{\tE} + \tau\cdot\abs{V} +
\abs{V}^2 \cdot (1.2738^p + p \cdot \abs{V}))$.\qed
\end{proof}
%}%appendixproof

Based on \cref{thm:noway,thm:compute-switching}, we can now show a fixed-parameter
algorithm for the parameter lifetime $\tau$ combined with the parameter
number~$p$ of switching vertices apart from~$s, z$.

\begin{theorem}
Let $\tG$ be a temporal split graph with at most $p$ switching vertices apart from~$s$ and~$z$.
Then \TS{} on $\tG$ can be solved
in~$\bigO( (\tau + 1)^{3^p(p+1)}\abs{\tG} + 1.2738^p \cdot \abs{V}^2 + p \cdot \abs{V}^3 )$ time.
\end{theorem}

\begin{proof}
We begin with using \cref{thm:compute-switching} to compute a temporal split partition of~$\tG$.
We then use induction to show that for all values of~$p$ there are at most~$(\tau+1)^{3^p(p+1)}$ minimal temporal $(s,z)$-separators,
which can all be found in~$D(\tau+1)^{3^p(p+1)}\abs{\tG}$ time for some constant~$D$.

For the case~$p=0$, by \cref{thm:noway}, we can find all minimal separators
of which there are at most~$\tau+1$, in $\bigO(\abs{\tG}\cdot\tau)$~time.

Now for the induction step suppose that our claim holds whenever the number of switching vertices (apart from $s,z$) is at most~$p-1$.
We choose a switching vertex~$v$ from~$\tG$ ($v \notin \{s, z\}$).
The subgraph~$\tG-v$ then contains~$p-1$ switching vertices apart from $s,z$,
therefore we can find all its minimal separators in~$D(\tau+1)^{3^{p-1}p}\abs{\tG}$ time.
Since a separator of~$\tG$ is also a separator of~$\tG-v$, every minimal separator of~$\tG$ must contain a separator of~$\tG-v$.
We will base our separators for~$\tG$ on the minimal separators of~$\tG-v$, henceforth called \emph{base separators},
by finding all possible combinations of vertices that can be added to turn them into minimal separators of~$\tG$.

Because we only added~$v$, all temporal ($s,z$)-paths left in~$\tG$ after removing a separator of~$\tG-v$ must pass through~$v$.
Thus any separator of $\tG$  must either contain~$v$ or some other set of vertices that cuts all these paths.
To do the latter, it has to ensure that all remaining temporal $(s,v)$-paths arrive after the latest layer in which a temporal $(v,z)$-path can begin.
In other words, such a separator of~$\tG$ needs to contain a temporal ($s,v$)-separator for the layers from~$1$ to some layer~$t$, and a temporal ($v,z$)-separator from~$t+1$ to~$\tau$.
We can compute all minimal separators of this form by applying the induction hypothesis
to enumerate all temporal $(s,v)$-separators in layers $1$~through~$t$ of~$\tG - \{z\}$
and all temporal $(v,z)$-separators in layers $t+1$~through~$\tau$ of $\tG - \{s\}$.
Note that both of these are temporal split graphs with at most $p-1$ switching vertices (not counting $s$, $v$, and $z$).

So for any given~$t$, there are no more than~$((\tau+1)^{3^{p-1}p})^2$ possible separator combinations.
Additionally we have the option of simply taking~$v$.
As there are~$\tau+1$ options for~$t$ and~$(\tau+1)^{3^{p-1}p}$ base separators to choose from,
the overall number of minimal temporal $(s,z)$-separators is thus upper-bounded by
\begin{align*}
	&(\tau+1)^{3^{p-1}p} \left( \left((\tau+1)^{3^{p-1}p}\right)^2(\tau+1)+1\right)
	\leq (\tau+1)^{3^{p}(p+1)}.
\end{align*}
We require $D(\tau+1)^{3^{p-1}p}\abs{\tG}$~time to find all base separators.
In addition, for each $t \in \{0, \dots, \tau\}$ we need $2D(\tau+1)^{3^{p-1}p} \abs{\tG}$~time to compute all minimal $(s,v)$- and $(v,z)$-separators.
Afterwards we need $\bigO(\abs{V})$ time to output each found separator.
The overall time required thus is
\begin{align*}
	&D(\tau+1)^{3^{p-1}p}\abs{\tG} \cdot (1 + 2(\tau+1))
	+ D(\tau+1)^{3^{p-1}2p} \abs{V}
	\\ \leq{}
	&D(\tau+1)^{3^p (p+1)}\abs{\tG}
\end{align*}
where we assumed $\tau \geq 3$.
This completes the induction.
Together with the time required for~\cref{thm:compute-switching},
we obtain an overall upper time bound of
\begin{align*}
&\bigO\left( (\tau + 1)^{3^p(p+1)} \abs{\tG} 
+ \abs{\tE} + \tau \cdot \abs{V} + \abs{V}^2 \cdot (1.2738^p + p\cdot \abs{V}) \right)
\\ \subseteq {} &
\bigO\left( (\tau + 1)^{3^p(p+1)}\abs{\tG} + 1.2738^p \cdot \abs{V}^2 + p \cdot
\abs{V}^3 \right) .
\end{align*}\qed
\end{proof}
We leave open whether for the single parameter number~$p$
of switching vertices \TS{} is \fpt{}.

%%%%%%%%%%%%%%%%%%%%%%%%%%%%PERMUTATION GRAPHS%%%%%%%%%%%%%%%%%%%%%%%%%%%%%%%%%%%%%

\section{Permutation Graphs}\label{chap:perm}
In this section, we investigate the complexity of finding temporal
separators in \emph{temporal permutation graphs}.
A permutation graph is defined by a ordered set of vertices (say $1, \dots, n$)
and any permutation of that vertex set. Two
vertices are connected by an edge if and only if their relative order is inverted by the permutation.
%\todo[inline]{Motivation} 

Formally,
we call a graph~$G=(V,E)$ with vertex set~$V = [n] := \{1, 2, \dots, n\}$
a \emph{permutation graph},
if there exists a permutation~$\pi : V \rightarrow V$ of the vertices
such that for any two vertices $v < w$, 
we have $\{v,w\} \in E$ if and only if $\pi(w)< \pi(v)$.
Clearly,~$\pi$ is then uniquely determined by~$G$.
Note that we follow the definition of \citet{EvenPermutation72} in which the vertices are already labeled;
some authors also define a permutation graph as any (unlabeled) graph for which such a labeling can be found~\cite{Gol04}.

We call a temporal graph~$\tG=([n],\tE,\tau)$ a \emph{temporal permutation graph}
if there exist $\tau$ permutations~$\pi_1, \ldots, \pi_\tau$ of
the vertices such that, for any two vertices $v < w$,
we have $(\{v,w\},t) \in \tE$ if and only if $\pi_t(w)<\pi_t(v)$.

One can visualize a permutation with a \emph{matching diagram} \cite{EvenPermutation72}.
A matching diagram for a given permutation graph $G=([n],E)$ with permutation~$\pi$ can be constructed by drawing~$n$ points on a
horizontal line and another~$n$ points on a parallel line below it.
Then each vertex~$i$ is represented by a straight line segment
connecting the $i$-th point on the top line to the $\pi(i)$-th point on the bottom line.
Observe that two vertices share an edge in~$G$ if and only if their corresponding line segments cross in the matching diagram.
\cref{fig:ex_per} provides an example for a permutation graph and the matching 
diagram of its underlying permutation. %is provided in \cref{fig:ex_per}.

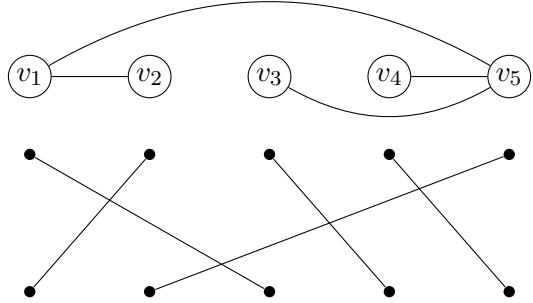
\begin{figure}[t]
	\centering
 	\tikzstyle{alter}=[circle, minimum size=16pt, draw, inner sep=1pt]
	\tikzstyle{dot}=[circle, fill, inner sep=1.5pt]
	\begin{tikzpicture}
		\node[alter] at (0,0) (v1) {$v_1$};
		\node[alter, right = of v1] (v2) {$v_2$};
		\node[alter, right = of v2] (v3) {$v_3$};
		\node[alter, right = of v3] (v4) {$v_4$};
		\node[alter, right = of v4] (v5) {$v_5$};
		
		\draw (v1) edge (v2);
		\draw (v1) [bend left=30] edge (v5);
		\draw (v3) [bend right=30] edge (v5);
		\draw (v4) edge (v5);
		
		\node[dot, below = 4ex of v1] (u1) {};
		\node[dot, below = 4ex of v2] (u2) {};
		\node[dot, below = 4ex of v3] (u3) {};
		\node[dot, below = 4ex of v4] (u4) {};
		\node[dot, below = 4ex of v5] (u5) {};
		
		\node[dot, below = 10ex of u1] (l1) {};
		\node[dot, below = 10ex of u2] (l2) {};
		\node[dot, below = 10ex of u3] (l3) {};
		\node[dot, below = 10ex of u4] (l4) {};
		\node[dot, below = 10ex of u5] (l5) {};
		
		\draw (u1) edge (l3);
		\draw (u2) edge (l1);
		\draw (u3) edge (l4);
		\draw (u4) edge (l5);
		\draw (u5) edge (l2);
	\end{tikzpicture}
	\caption{An example of a permutation graph and the corresponding matching diagram.}
	\label{fig:ex_per}
\end{figure}

First, we prove in \cref{chap:perm:hardness} that \TS{} on temporal permutation
graphs is NP-complete.
In \cref{chap:perm:swap}, we use the \emph{Kendall tau} distance~\cite{Ken38} to measure
the dissimilarity of the permutations. Recall that the Kendall tau distance is a metric
that counts the number of pairwise disagreements between two total orderings; it is also known as ``bubble sort distance'' since it measures the number of swaps needed to transform one permutation into the other. 
We show that \TS{} becomes \fpt{} when parameterized by
the combined parameter ``sum of Kendall tau distances between consecutive permutations'' and the separator size~$k$.

\subsection{Hardness Results}\label{chap:perm:hardness}

As is the case for split graphs, the class of permutation graphs contains all graphs with only one edge. This means that
\TS{} for temporal graphs in which every layer is a permutation graph is NP-hard \cite{Flu20}. We can also show that \TS{} remains NP-hard when restricted
to the temporal permutation graphs (note that this means that the vertex
ordering is the same for all layers).

\begin{theorem}%[\appprooflink{thm:const}]
\label{thm:const}
\TS{} on temporal permutation graphs is NP-complete.
\end{theorem}
%\appendixproof{thm:const}{
\begin{proof}
The problem is clearly contained in NP. We present a polynomial-time reduction
from 3-SAT~\cite{GJ79} to an instance of \TS{} on a temporal permutation graph.
In 3-SAT we are given a Boolean formula in conjunctive normal form with at most
three literals in a clause and are asked whether there exists a truth assignment
that satisfies the formula.

The constructed temporal graph contains two vertices for every unique variable in the original 3-SAT instance,
exactly one of which has to be included in the separator, representing the choice of assigning a Boolean value to that variable.
It also contains three vertices for every clause in the 3-SAT instance, one vertex per literal. Of these,
exactly two per clause must be included in the separator. For each literal, there exists a temporal ($s,z$)-path that consists of
that literal's vertex and a variable vertex that is part of the separator if and only if the literal evaluates to true.
Since we can only remove two of a clause's three vertices, one of its literals needs to evaluate to true in order to avoid any
remaining temporal ($s,z$)-paths.

Let~$\var:=\{x_1,x_2,...,x_{n}\}$ be the set of variables of a given 3-SAT instance with the
Boolean formula~$(\ell_{1,1}\lor \ell_{1,2}\lor \ell_{1,3})\land ...\land
(\ell_{m,1}\lor \ell_{m,2}\lor \ell_{m,3})$.
We can construct a \TS{} instance~$\mathcal{I}:=(\tG=(V,\tE,n+5m+2),s,z,k=n+2m)$,
where~$V= [2n + 3m + 2]$.
For the sake of readability, we denote the vertices as follows:
\begin{align*}
	s &:= 1 &&\\
	t_i &:= 2i  &&\forall i \in [n]\\
	f_i &:= 2i+1 &&\forall i \in [n]\\
	c_{i,j} &:=1+2n+3(i-1)+j &&\forall i \in [m], j \in [3]\\
	z &:= 2n + 3m + 2 &&
\end{align*}
% , such that there exists a function~$\nu:V\rightarrow\NN$ that is a valid naming for all layers of~$\tG$.
We define the time-edges as
\begin{align*}
	E_i(\tG) :={}& \{\{t_i,v\},\{f_i,w\}\mid v,w\in V, t_i< v, w< f_i\},\\
	&\text{for } i\in\NN, i\leq n,\\
	E_{n+1}(\tG) :={}& \{\{s,t_i\},\{s,f_i\}\mid i\in\NN, i\leq n\},\\
	E_{n+1+3(i-1)+j}(\tG) :={}& \{\{v,c_{i,j}\}\}\cup {}\\
	&\{\{v,w\}\mid w\in V, v< w\leq f_{n}\}\cup {} \\
	&\{\{c_{i,j},w\}\mid w\in V, c_{1,1}\leq w< c_{i,j}\},\\
	&\text{where } v=t_k \Leftrightarrow l_{i,j}=x_k \text{ and } v=f_k \Leftrightarrow l_{i,j}=\neg x_k,\\
	&\text{for } i,j\in\NN, i\leq m, j\leq 3,\\
	E_{n+3m+2}(\tG) :={}& \{\{z,c_{i,j}\}\mid i,j\in\NN, i\leq m, j\leq 3 \},\\
	E_{n+3m+2i+1}(\tG) :={}& \{\{c_{i,2},v\},\{c_{i,3},w\}\mid v,w\in V,
	c_{i,2}< v, w< c_{i,3}\},\\
	&\text{for } i\in\NN, i\leq m,\\
	E_{n+3m+2i+2}(\tG) :={}& \{\{c_{i,1},v\},\{c_{i,2},w\},\{c_{i,3},w\}\mid v,w\in V, c_{i,1}< v, w< c_{i,2}\},\\
	&\text{for } i\in\NN, i\leq m.
\end{align*}

\begin{figure}[p]
	\centering
 	\tikzstyle{alter}=[circle, minimum size=16pt, draw, inner sep=1pt]
	\tikzstyle{dot}=[circle, fill, inner sep=1.5pt]
	\begin{tikzpicture}
		\node[dot] at (0,0) (u1) {};
		\node[right = of u1] (e1) {$...$};
		\node[dot, right = of e1] (u2) {};
		\node[dot, right = of u2] (u3) {};
		\node[dot, right = of u3] (u4) {};
		\node[dot, right = of u4] (u5) {};
		\node[right = of u5] (e2) {$...$};
		\node[dot, right = of e2] (u6) {};
	
		\node[above = 1ex of u1] {$s$};
		\node[above = 1ex of u2] {$f_{i-1}$};
		\node[above = 1ex of u3] {$t_i$};
		\node[above = 1ex of u4] {$f_i$};
		\node[above = 1ex of u5] {$t_{i+1}$};
		\node[above = 1ex of u6] {$z$};
		
		%\draw (v3) edge (v4);
		%\draw (v3) [bend left=20] edge (v5);
		%\draw (v3) [bend left=30] edge (v6);
		%\draw (v4) [bend left=20] edge (v2);
		%\draw (v4) [bend left=30] edge (v1);
		
		\node[dot, below = 10ex of u1] (l1) {};
		\node[dot, right = of l1] (l2) {};
		\node[right = of l2] (e3) {$...$};
		\node[dot, right = of e3] (l3) {};
		\node[dot, right = of l3] (l4) {};
		\node[right = of l4] (e4) {$...$};
		\node[dot, right = of e4] (l5) {};
		\node[dot, right = of l5] (l6) {};
		
		\draw (u1) edge (l2);
		\draw (u2) edge (l3);
		\draw (u3) edge (l6);
		\draw (u4) edge (l1);
		\draw (u5) edge (l4);
		\draw (u6) edge (l5);
	\end{tikzpicture}
	\caption{Matching diagram for some layer~$i\in\NN, 1<i<n$, of~$\tG$ in the proof of \cref{thm:const}.}
	\label{fig:per1}
\end{figure}
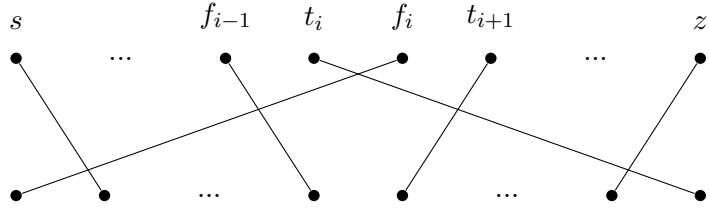

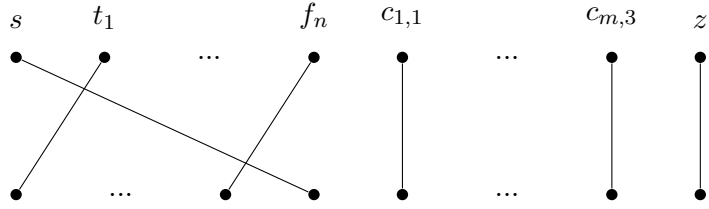
\begin{figure}[p]
	\centering
 	\tikzstyle{alter}=[circle, minimum size=16pt, draw, inner sep=1pt]
	\tikzstyle{dot}=[circle, fill, inner sep=1.5pt]
	\begin{tikzpicture}
		\node[dot] at (0,0) (u1) {};
		\node[dot, right = of u1] (u2) {};
		\node[right = of u2] (e1) {$...$};
		\node[dot, right = of e1] (u3) {};
		\node[dot, right = of u3] (u4) {};
		\node[right = of u4] (e2) {$...$};
		\node[dot, right = of e2] (u5) {};
		\node[dot, right = of u5] (u6) {};
	
		\node[above = 1ex of u1] {$s$};
		\node[above = 1ex of u2] {$t_1$};
		\node[above = 1ex of u3] {$f_{n}$};
		\node[above = 1ex of u4] {$c_{1,1}$};
		\node[above = 1ex of u5] {$c_{m,3}$};
		\node[above = 1ex of u6] {$z$};
		
		\node[dot, below = 10ex of u1] (l1) {};
		\node[right = of l1] (e3) {$...$};
		\node[dot, right = of e3] (l2) {};
		\node[dot, right = of l2] (l3) {};
		\node[dot, right = of l3] (l4) {};
		\node[right = of l4] (e4) {$...$};
		\node[dot, right = of e4] (l5) {};
		\node[dot, right = of l5] (l6) {};
		
		\draw (u1) edge (l3);
		\draw (u2) edge (l1);
		\draw (u3) edge (l2);
		\draw (u4) edge (l4);
		\draw (u5) edge (l5);
		\draw (u6) edge (l6);
	\end{tikzpicture}
	\caption{Matching diagram for layer~$n+1$ of~$\tG$ in the proof of \cref{thm:const}.}
	\label{fig:per2}
\end{figure}

\begin{figure}[p]
	\centering
 	\tikzstyle{alter}=[circle, minimum size=16pt, draw, inner sep=1pt]
	\tikzstyle{dot}=[circle, fill, inner sep=1.5pt]
	\begin{tikzpicture}
		\node[dot] at (0,0) (u1) {};
		\node[right = 1ex of u1] (e1) {$...$};
		\node[dot, right = 1ex of e1] (u2) {};
		\node[dot, right = of u2] (u3) {};
		\node[dot, right = of u3] (u4) {};
		\node[right = 1ex of u4] (e2) {$...$};
		\node[dot, right = 1ex of e2] (u5) {};
		\node[dot, right = of u5] (u6) {};
		\node[right = 1ex of u6] (e3) {$...$};
		\node[dot, right = 1ex of e3] (u7) {};
		\node[dot, right = of u7] (u8) {};
		\node[dot, right = of u8] (u9) {};
		\node[right = 1ex of u9] (e4) {$...$};
		\node[dot, right = 1ex of e4] (u10) {};
	
		\node[above = 1ex of u1] {$s$};
		\node[above = 1ex of u2] {$f_{k-1}$};
		\node[above = 1ex of u3] {$t_k$};
		\node[above = 1ex of u4] {$f_k$};
		\node[above = 1ex of u5] {$f_{n}$};
		\node[above = 1ex of u6] {$c_{1,1}$};
		\node[above = 1ex of u7] {$c_{i,1}$};
		\node[above = 1ex of u8] {$c_{i,2}$};
		\node[above = 1ex of u9] {$c_{i,3}$};
		\node[above = 1ex of u10] {$z$};
		
		\node[dot, below = 10ex of u1] (l1) {};
		\node[right = 1ex of l1] (e5) {$...$};
		\node[dot, right = 1ex of e5] (l2) {};
		\node[dot, right = of l2] (l3) {};
		\node[right = 1ex of l3] (e6) {$...$};
		\node[dot, right = 1ex of e6] (l4) {};
		\node[dot, right = of l4] (l5) {};
		\node[dot, right = of l5] (l6) {};
		\node[dot, right = of l6] (l7) {};
		\node[right = 1ex of l7] (e7) {$...$};
		\node[dot, right = 1ex of e7] (l8) {};
		\node[dot, right = of l8] (l9) {};
		\node[right = 1ex of l9] (e8) {$...$};
		\node[dot, right = 1ex of e8] (l10) {};
		
		\draw (u1) edge (l1);
		\draw (u2) edge (l2);
		\draw (u3) edge (l6);
		\draw (u4) edge (l3);
		\draw (u5) edge (l4);
		\draw (u6) edge (l7);
		\draw (u7) edge (l8);
		\draw (u8) edge (l5);
		\draw (u9) edge (l9);
		\draw (u10) edge (l10);
	\end{tikzpicture}
	\caption{Matching diagram for layer~$n+1+3(i-1)+2, i\in\NN, 1<i<m$, of~$\tG$ in the proof of \cref{thm:const}, assuming that~$l_{i,2}=x_k$.}
	\label{fig:per3}
\end{figure}
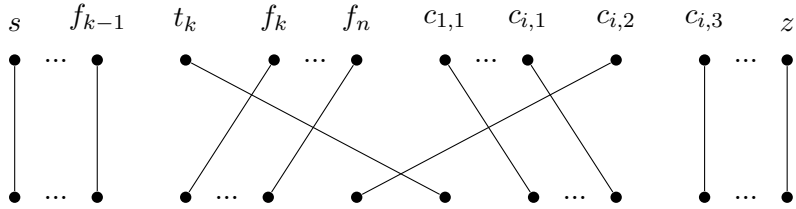

\begin{figure}[p]
	\centering
 	\tikzstyle{alter}=[circle, minimum size=16pt, draw, inner sep=1pt]
	\tikzstyle{dot}=[circle, fill, inner sep=1.5pt]
	\begin{tikzpicture}
		\node[dot] at (0,0) (u1) {};
		\node[dot, right = of u1] (u2) {};
		\node[right = of u2] (e1) {$...$};
		\node[dot, right = of e1] (u3) {};
		\node[dot, right = of u3] (u4) {};
		\node[right = of u4] (e2) {$...$};
		\node[dot, right = of e2] (u5) {};
		\node[dot, right = of u5] (u6) {};
	
		\node[above = 1ex of u1] {$s$};
		\node[above = 1ex of u2] {$t_1$};
		\node[above = 1ex of u3] {$f_{n}$};
		\node[above = 1ex of u4] {$c_{1,1}$};
		\node[above = 1ex of u5] {$c_{m,3}$};
		\node[above = 1ex of u6] {$z$};
		
		\node[dot, below = 10ex of u1] (l1) {};
		\node[dot, right = of l1] (l2) {};
		\node[right = of l2] (e3) {$...$};
		\node[dot, right = of e3] (l3) {};
		\node[dot, right = of l3] (l4) {};
		\node[dot, right = of l4] (l5) {};
		\node[right = of l5] (e4) {$...$};
		\node[dot, right = of e4] (l6) {};
		
		\draw (u1) edge (l1);
		\draw (u2) edge (l2);
		\draw (u3) edge (l3);
		\draw (u4) edge (l5);
		\draw (u5) edge (l6);
		\draw (u6) edge (l4);
	\end{tikzpicture}
	\caption{Matching diagram for layer~$n+3m+2$ of~$\tG$ in the proof of \cref{thm:const}.}
	\label{fig:per4}
\end{figure}

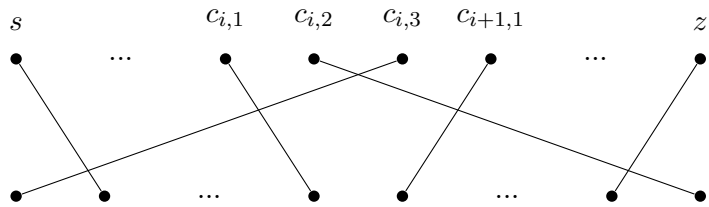
\begin{figure}[p]
	\centering
 	\tikzstyle{alter}=[circle, minimum size=16pt, draw, inner sep=1pt]
	\tikzstyle{dot}=[circle, fill, inner sep=1.5pt]
	\begin{tikzpicture}
		\node[dot] at (0,0) (u1) {};
		\node[right = of u1] (e1) {$...$};
		\node[dot, right = of e1] (u2) {};
		\node[dot, right = of u2] (u3) {};
		\node[dot, right = of u3] (u4) {};
		\node[dot, right = of u4] (u5) {};
		\node[right = of u5] (e2) {$...$};
		\node[dot, right = of e2] (u6) {};
	
		\node[above = 1ex of u1] {$s$};
		\node[above = 1ex of u2] {$c_{i,1}$};
		\node[above = 1ex of u3] {$c_{i,2}$};
		\node[above = 1ex of u4] {$c_{i,3}$};
		\node[above = 1ex of u5] {$c_{i+1,1}$};
		\node[above = 1ex of u6] {$z$};
		
		\node[dot, below = 10ex of u1] (l1) {};
		\node[dot, right = of l1] (l2) {};
		\node[right = of l2] (e3) {$...$};
		\node[dot, right = of e3] (l3) {};
		\node[dot, right = of l3] (l4) {};
		\node[right = of l4] (e4) {$...$};
		\node[dot, right = of e4] (l5) {};
		\node[dot, right = of l5] (l6) {};
		
		\draw (u1) edge (l2);
		\draw (u2) edge (l3);
		\draw (u3) edge (l6);
		\draw (u4) edge (l1);
		\draw (u5) edge (l4);
		\draw (u6) edge (l5);
	\end{tikzpicture}
	\caption{Matching diagram for some layer~$n+3m+2i+1, i\in\NN, 1<i<m$, of~$\tG$ in the proof of \cref{thm:const}.}
	\label{fig:per5}
\end{figure}

\begin{figure}
	\centering
 	\tikzstyle{alter}=[circle, minimum size=16pt, draw, inner sep=1pt]
	\tikzstyle{dot}=[circle, fill, inner sep=1.5pt]
	\begin{tikzpicture}
		\node[dot] at (0,0) (u1) {};
		\node[right = of u1] (e1) {$...$};
		\node[dot, right = of e1] (u2) {};
		\node[dot, right = of u2] (u3) {};
		\node[dot, right = of u3] (u4) {};
		\node[dot, right = of u4] (u5) {};
		\node[dot, right = of u5] (u6) {};
		\node[right = of u6] (e2) {$...$};
		\node[dot, right = of e2] (u7) {};
	
		\node[above = 1ex of u1] {$s$};
		\node[above = 1ex of u2] {$c_{i-1,3}$};
		\node[above = 1ex of u3] {$c_{i,1}$};
		\node[above = 1ex of u4] {$c_{i,2}$};
		\node[above = 1ex of u5] {$c_{i,3}$};
		\node[above = 1ex of u6] {$c_{i+1,1}$};
		\node[above = 1ex of u7] {$z$};
		
		\node[dot, below = 10ex of u1] (l1) {};
		\node[dot, right = of l1] (l2) {};
		\node[dot, right = of l2] (l3) {};
		\node[right = of l3] (e3) {$...$};
		\node[dot, right = of e3] (l4) {};
		\node[dot, right = of l4] (l5) {};
		\node[right = of l5] (e4) {$...$};
		\node[dot, right = of e4] (l6) {};
		\node[dot, right = of l6] (l7) {};
		
		\draw (u1) edge (l3);
		\draw (u2) edge (l4);
		\draw (u3) edge (l7);
		\draw (u4) edge (l1);
		\draw (u5) edge (l2);
		\draw (u6) edge (l5);
		\draw (u7) edge (l6);
	\end{tikzpicture}
	\caption{Matching diagram for some layer~$n+3m+2i+2, i\in\NN, 1<i<m$, of~$\tG$ in the proof of \cref{thm:const}.}
	\label{fig:per6}
\end{figure}
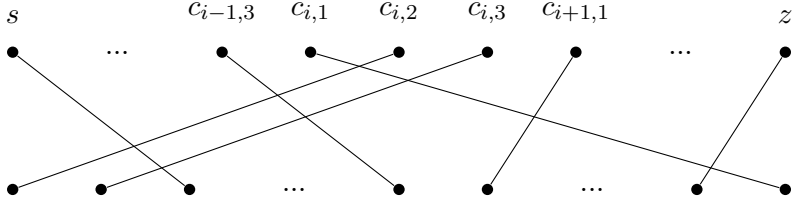

We first show that we can give a
permutation for every layer that produces the corresponding
edge set.

In layers~1 through~$n$
and all layers~$n+3m+2i+1, i\in\NN$, there is one vertex
with edges to all smaller vertices and one with edges
to all larger vertices. The corresponding permutation maps these vertices to the
first respectively last position while keeping the relative order of all other vertices unchanged.
If said two vertices are~$i$ and~$i+1$, then the corresponding permutation would be~$(2,3,...,i,|V|,1,i+1,...,|V|-1)$.
See also \cref{fig:per1} and \cref{fig:per5}.

In layer~$n+1$, vertex~$s$ has edges to the set~$\{2, 3, \dots, f_n\}$,
yielding the permutation~$(f_n,1,2,...,$ $f_n-1,f_n+1,...,|V|)$.
For layer~$n+3m+2$, the situation is similar.
See also \cref{fig:per2} and \cref{fig:per4}.

In each of the layers $n+2$~through~$n+3m+1$ there is a vertex~$v\leq 2n+1$ and a vertex~$c>2n+1$.
Each of these two vertices has an edge to every other vertex lying between themselves and $f_n + 1/2$.
This corresponds to the permutation~$(1,2,...,v-1,f_n+1,v,...,f_n-1,f_n+2,...,c,f_n,c+1,...,|V|)$.
See \cref{fig:per3}.

In all layers~$n+3m+2i+2, i\in\NN$, there are two consecutive vertices with
edges to all vertices below them but no edge to each
other, and one vertex with edges to all vertices above it. This
corresponds to the positions of the former two vertices having the first two
positions in the permutation and the latter vertex having the last position in
the permutation while the order of all other elements remains unchanged. For example, for 
layer~$n+3m+2i+2$, this yields the
permutation~$(3,4,...,c_{i,2},|V|,1,2,c_{i,3},...,|V|-1)$. See \cref{fig:per6}.

\paragraph*{Correctness.}
We prove that the Boolean formula is satisfiable if and only if~$\mathcal{I}$ is a yes-instance.

$\Leftarrow$: For every~$i\in\NN, i\leq n$, there exists an ($s,z$)-path through~$t_i$ and~$f_i$ in layer~$i$, meaning that one
of these two vertices has to be included in any temporal ($s,z$)-separator. For every~$i\in\NN, i\leq m$,
there is an ($s,z$)-path through~$c_{i,2}$ and~$c_{i,3}$ in layer~$n+3m+2i+1$ and two paths in layer~$n+3m+2i+2$,
first going through~$c_{i,1}$ and then either~$c_{i,2}$ or~$c_{i,3}$. To remove all of these paths any two of the three vertices
have to be in any temporal ($s,z$)-separator.

To satisfy all of these requirements, a temporal ($s,z$)-separator needs to contain
at least~$n+2m$ vertices. Since this equals~$k$, a separator that is a solution to~$\mathcal{I}$
cannot contain any more vertices than required above.
Specifically, for each~$i\in\NN, i\leq n$, only one of the two vertices~$t_i$ and~$f_i$ can be in the separator, and
for every clause~$(\ell_{i,1}\lor \ell_{i,2}\lor \ell_{i,3})$ in the Boolean formula,
there exists one vertex~$c_{i,j}, j\in\{1,2,3\}$ that is not in the separator. If~$\ell_{i,j}$ is of the form~$x_k$, then setting~$x_k$ to true satisfies the clause and
there exists an ($s,z$)-path through~$t_k$ and~$c_{i,j}$ spanning layer~$n+1$, layer~$n+1+3(i-1)+j$, and
layer~$n+3m+2$, meaning that~$t_k$ must be in the separator.
If~$\ell_{i,j}$ is of the form~$\neg x_k$, then setting~$x_k$ to false satisfies the clause and
there exists an ($s,z$)-path through~$f_k$ and~$c_{i,j}$ in those three layers, meaning that~$f_k$ must be in the separator.
Hence, if we set every variable~$v_k$ in~$\var$ to true if the separator contains~$t_k$, and to false if the separator
contains~$f_k$, all clauses, and therefore the Boolean formula as a whole, are satisfied.

$\Rightarrow$: Given an interpretation of~$\var$ that satisfies the Boolean formula, we can construct a temporal ($s,z$)-separator~$S$
of size~$k$ for~$\tG$. For~$i\in\NN, i\leq n$, if the interpretation sets~$x_i$ to true, then we include~$t_i$ in the separator,
if the interpretation set~$x_i$ to false, then we include~$f_i$ in the separator.
Since the interpretation satisfies the Boolean formula, in every
clause~$(\ell_{i,1}\lor \ell_{i,2}\lor \ell_{i,3})$ in the formula at least one literal has to evaluate to true. We choose~$\ell_{i,j}$ to be such a literal and then include the vertices
in~$\{c_{i,k}\in V\mid k\neq j\}$ in the separator. After doing this for each clause, we now have a temporal ($s,z$)-separator
of size~$k$.

In the following we will refer to the vertex set~$\{v\in V\mid v \leq f_{n}\}$ as~$V_v$,
and to~$\{v\in V\mid c_{1,1} \leq v\}$ as~$V_c$.
In each layer~$i$ for~$i\in\NN, i\leq n$, only vertices up to~$f_i$ are reachable from~$s$ without having to
pass through both~$t_i$ and~$f_i$, so by layer~$n$, only vertices in~$V_v$ can be reached. In layer~$n+1$,
there are no edges leading to~$V_c$.

In each layer~$n+1+3(i-1)+j$ for~$i,j\in\NN, i\leq m, j\leq 3$, the only edge
connecting~$V_v$ and~$V_c$ is~$\{v, c_{i,j}\}$ where~$v=t_k\Leftrightarrow \ell_{i,j}=x_k$
and~$v=f_k\Leftrightarrow \ell_{i,j}=\neg x_k$, with~$k\in\NN, k\leq n$. We look at two cases:
\begin{itemize}
	\item If~$c_{i,j}$ is in~$S$, then~$V_c$ cannot be reached in this layer.
	\item If we did not include~$c_{i,j}$ in~$S$, then~$\ell_{i,j}$ must evaluate to true, now we distinguish two cases again:
	\begin{itemize}
		\item If~$\ell_{i,j}=x_k$, then~$x_k$ must be set to true and~$v=t_k\in S$.
		\item If~$\ell_{i,j}=\neg x_k$, then~$x_k$ must be set to false and~$v=f_k\in S$.
	\end{itemize}
\end{itemize}
Therefore, in none of these layers any vertex in~$V_c$ is reachable from~$V_v$.

In layer~$n+3m+2$, there are no edges leading to~$V_c$.
In the layers~$n+3m+2i+1$ and~$n+3m+2i+2$ for each~$i\in\NN, i\leq m$, no
vertices greater than~$c_{i,3}$ can be reached
from~$V_v$ without passing through at least two of the vertices~$c_{i,1},c_{i,2}$, and~$c_{i,3}$. Since only one of these was not included in~$S$, this is not possible.

Hence, by layer~$n+5m+2$, which is the last layer, no vertex greater than~$2n+3n+1$ can be reached from~$s$.
Since~$z=\abs{V}=2n+3m+2$, there remains no temporal ($s,z$)-path
and~$S$ is indeed a temporal ($s,z$)-separator.\qed
\end{proof}
%}%appendixproof

We remark that the reduction we use to obtain \cref{thm:const} uses a unbounded number of time steps and also the maximum Kendall tau distance between any two consecutive permutations is unbounded. However, by introducing additional layers, one can decrease the Kendall tau distance between any two consecutive layers to one. The main idea is to gradually change a layer to an independent set and then gradually to the next layer. This can be done in a way that does not introduce any new temporal paths.

\subsection{Fixed-Parameter Tractability Results}\label{chap:perm:swap}

In this section, we examine the effect of limiting how much the permutations
of the layers of the temporal permutation graph change. We can do so by
measuring the Kendall tau distance between the permutations of consecutive layers.
% The swap distance of two permutations is defined by how many times we need to swap the positions of two adjacent elements
%to get from one permutation to the other.

We will show that \TS{} on temporal permutation graphs is \fpt{} with respect to the sum of Kendall tau distances between all pairs of consecutive permutations plus separator size~$k$. For this we need to demonstrate that these parameters limit the number of
($s,z$)-separators that a layer does not have in common with another layer.

First, we introduce the concept of \emph{scanlines}.
A scanline is any line segment in the matching diagram of a permutation with one end on each horizontal line.
If~$s$ lies on one side of the scanline and~$z$ on the other, then the set of all vertices whose line segments
cross the scanline is an ($s,z$)-separator. We call such a separator a \emph{scanline separator}.
\citet{BKK95} have shown that every minimal separator in a permutation graph is
a scanline separator and that there are at most $(n-1)(2k+3)$~scanline separators of size at most~$k$ \cite[Proof of Lemma~3.6]{BKK95}, where $n$ is the number of vertices.

\begin{lemma} \label{thm:dk_scan}
Let $G_1=([n],E_1)$ and~$G_2=([n],E_2)$ be two permutation graphs.
If the two corresponding permutations have Kendall tau
distance~$d$, then the number of scanline separators of size at
most~$k$ in~$G_2$ that are not also scanline separators
in~$G_1$ is at most~$d\cdot(2k+1)$.
\end{lemma}

\begin{proof}
We begin by showing that for every endpoint on the bottom line, there are at most~$2k+1$ separators of size at most~$k$
defined by a scanline with that endpoint. For a given bottom endpoint, we look at the leftmost endpoint on the top line that
defines such a separator. Moving right, every time we pass an element of the permutation, its line is added to those crossed by
the scanline if it was completely to the right of it before, and removed from them if it crossed the scanline before.
This corresponds to adding or removing one vertex from the defined separator.

Since our scanline can pass every element only once,
the vertices that were added in this process will not be removed again.
If we have already passed more than $2k$ points on the top line,
then more than~$k$ of these represent vertices which were not present in the initial separator and were thus added.
Therefore the resulting separator contains more than~$k$ vertices.
Therefore, all separators of size at most~$k$ must be produced in the first $2k$ steps of this process.
Together with the separator defined by the
initial scanline, this gives us at most~$2k+1$ scanline separators as claimed.

Upon moving from~$G_1$ to~$G_2$,
the set of lines crossing any given scanline only changes,
if any of the $d$~swaps swapped the two points immediately to the left and right of the lower end of the scanline.
By the above, this means at most $d \cdot (2k+1)$ many scanlines produce different separators in $G_1$ and $G_2$.
Hence, the number of new
scanline separators in~$G_2$ does not exceed~$d\cdot(2k+1)$.\qed
\end{proof}

Now we show that reachability in temporal permutation graphs follows the vertex order:
If a vertex reaches another vertex, then it also reaches all vertices in between.

\begin{lemma}%[\appprooflink{thm:order}]
\label{thm:order}
Let~$G=([n],E)$ be a permutation graph 
and $v,w,x\in [n]$ three vertices with $v< w< x$.
If there exists a path from~$v$ to~$x$,
then there also exist paths from~$w$ to both, $v$~and~$x$.
\end{lemma}
%\todo{Reviewer 3 would prefer a citation over a proof.
%This actually holds for all co-comparability graphs and is probably related to AT-freeness. HM: should we put this into the appendix then? maybe instead add the proof of prop. 4?}
%\appendixproof{thm:order}{
\begin{proof}
Let $\pi$ be the corresponding permutation.
Since~$v$ and~$x$ are connected, it suffices to show that there is a path from one of them to~$w$.
We first look at the case that there is an edge between~$v$ and~$x$. This means
that~$\pi(v)>\pi(x)$.
If~$w$ does not have an edge to either~$v$ or~$x$, then neither the
positions of~$v$ and~$w$ nor~$w$ and~$x$ get reversed by $\pi$,
so~$\pi(v)<\pi(w)<\pi(x)\Rightarrow\pi(v)<\pi(x)$,
a contradiction. Therefore, either~$v$ or~$x$ must be connected to~$v$ via a single edge.

Now we use induction to show that the lemma holds for a ($v,x$)-path of any length. A ($v,x$)-path of length~1 is the same as
an edge between the two vertices. The induction hypothesis is that if there exists a ($v,x$)-path of length~$\ell\in\NN$, then there exist paths
to~$w$ from~$v$ and~$x$.

We show that if there exists a ($v,x$)-path~$p$ of length~$\ell+1$, then there exist paths to~$w$ from~$v$ and~$x$.
Let~$v'$ be the element of~$p$ that follows~$v$. If~$w< v'$, then, since~$v< w< v'$ and there exists an edge
between~$v$ and~$v'$, there is either an edge between~$v$ and~$w$, or an edge between~$v'$ and~$w$, forming a path from~$v$ to~$w$.
If~$v'< w$ then, since~$v'< w< x$ and there exists a ($v',x$)-path of length~$\ell$, due to the induction hypothesis, there
exists a path from~$x$ to~$w$.\qed
\end{proof}
%}%appendixproof

We now present a parameterized algorithm for \TS{} on a temporal permutation graph~$\tG$.
For this, we introduce the parameter $d_\Sigma := \sum_{t=1}^{\tau -1} d_\textnormal{Kt}(\pi_t, \pi_{t+1})$,
where $d_\textnormal{Kt}$ denotes the Kendall tau distance and $\pi_t$ is the $t$-th permutation of $\tG$. Note that taking the maximum instead of the sum does not provide a helpful parameter since the hardness reduction we used to obtain \cref{thm:const} can be modified such that the Kendall tau distance of any two consecutive layers is one.

\begin{theorem}
For a temporal permutation graph, \TS{} can be solved
in~$\bigO((d_\Sigma (2k+1))^k  n\cdot \abs{\tE} + \tau n^2)$ time.
\end{theorem}

\begin{proof}
We present an algorithm that runs in~$\bigO((d_\Sigma (2k+1))^k  n\cdot\abs{\tE} + \tau  n^2)$ time,
which determines whether a given \TS{}-instance has a solution (see \cref{alg:swap}).

\newcommand{\seplist}{\texttt{seplist}}
\newcommand{\layerseps}{\texttt{seps}}
\begin{algorithm}[t]
\caption{}\label{alg:swap}
\begin{algorithmic}[1]
\Input{A \TS{}-instance~$\mathcal{I}=(\tG=([n],\tE,\tau),s,z,k)$,
	where~$\tG$ is a temporal permutation graph and~$s,z\in [n]$}
\Output{\true{} if~$\mathcal{I}$ is a yes-instance, \false{} otherwise}
\State compute $\pi_1, \dots, \pi_\tau$ \label{line:init_first}
\State let \seplist{} be an empty list
\State append the set of scanline separators of size at most~$k$
in~$G_1(\tG)$ to \seplist{}
\For{$i\in\{2,...,\tau\}$}
	\State let \layerseps{} be the set of scanline separators of size at most~$k$
	in~$G_i(\tG)$ that are not scanline separators in~$G_{i-1}(\tG)$
	\If{\layerseps{} is not empty}
		\State append \layerseps{} to \seplist{}
	\EndIf
\EndFor \label{line:init_last}
\State output \Call{GetSeparator}{$\emptyset$, 1}
\vspace{0.5\baselineskip}
\Function{GetSeparator}{$S$, $i$}
	\If{$S$ is a temporal ($s,z$)-separator of~$\tG$} \label{line:yes_con}
		\State \textbf{return} $\true$
	\EndIf
	\For{$j \in \{i, \dots, \tau\}$}
		\For{$S'\in \seplist[j]$} \label{line:loop}
			\If{$\abs{S}<\abs{S\cup S'}\leq k$} \label{line:condition}
				\If{\Call{GetSeparator}{$S\cup S'$, $j+1$}}
					\State \textbf{return} $\true$
				\EndIf
			\EndIf
		\EndFor
	\EndFor
	\State \textbf{return} $\false$
\EndFunction
\end{algorithmic}
\end{algorithm}

Remember that the total number of scanline separators of size
at most~$k$ in layer~1 is at most~$(n-1)(2k+3)$.
Furthermore, in all layers after
layer~1, the number of minimal (i.e., scanline) separators which are not shared
with the previous layer is at most~$d_\Sigma (2k+1)$ (see \cref{thm:dk_scan}).
Hence the first call of \textproc{GetSeparator} iterates at most
$(n-1)(2k+3)+d_\Sigma (2k+1) \in \bigO\left( d_\Sigma (2k+1) \cdot n \right)$
times and
every recursive call iterates at most~$d_\Sigma (2k+1)$ times.

Due to the condition in \cref{line:condition}, every time a recursive call is made,
the set passed to~$S$ contains at least one vertex more than before, but never exceeds the size~$k$.
Thus the maximum recursion depth is~$k$. In every call of \textproc{GetSeparator} it is checked in~$\bigO(\abs{E})$ time whether~$S$
is a temporal ($s,z$)-separator. This results in a running time
of
$\bigO\left( (d_\Sigma(2k+1))^k  n \cdot \abs{\tE}\right)$
for the initial call of \textproc{GetSeparator}.

It remains to show that Lines \ref{line:init_first} through \ref{line:init_last} can be performed in $\bigO(\tau n^2)$~time.
To construct some permutation $\pi_i$, we first iterate once through the edges of $G_i$ and build, for each vertex $v \in [n]$,
the set $I(v) := \{w \in [n] \mid w < v \text{ and } \{v, w\} \in E(G_i)\}$.
Then we can incrementally construct $\pi_i$ from an empty tuple by going through all vertices in ascending order and inserting each vertex~$v$ exactly to the left of all elements of~$I(v)$.
If $I(v)$~is implemented using a hash set, this takes $\bigO(n^2)$~time for each layer.
Afterwards, building $\seplist{}$ again takes $\bigO(n^2)$~time for each layer as there are $n^2$ potential scanlines
and each only requires constant checking time if they are all iterated in order.

\paragraph*{Correctness.}
It is easy to see that \cref{alg:swap} will never output \true{} when~$\mathcal{I}$ is a no-instance,
as the condition in \cref{line:yes_con} can only evaluate to true if there exists a temporal ($s,z$)-separator.

It remains to be shown that \cref{alg:swap} will always output \true{} when~$\mathcal{I}$ is a yes-instance.
Without loss of generality we assume that~$s< z$. We also assume that there exists
a minimal temporal ($s,z$)-separator~$S^*$ of size at most~$k$ in~$\tG$.
Due to \cref{thm:order}, every layer~$t$ in~$\tG-S^*$ has some farthest reachable vertex~$f_t$ between~$s$ and~$z$
such that until time~$t$, $s$~can reach all vertices~$v$ with $s\leq v\leq f_t$ via temporal paths
but no vertex $v$ with $v > f_t$.
Clearly $f_t \leq f_{t+1}$.
This means that in each layer~$t$, $S^*$ must contain a scanline separator that separates~$f_t$ from all vertices $v > f_t$.
We denote this scanline separator by $S^*_t$.
By minimality of~$S^*$,
$S^* = \bigcup_{t=1}^\tau S^*_t$.

Trivially, the following property holds for the initial call of \textproc{GetSeparator}:
\[
\bigcup_{t=1}^{i-1} S^*_{t} \subseteq S \subseteq S^* \tag{$*$}
\]

We next show that whenever $(*)$ holds for a call of \textproc{GetSeparator}, then either $S = S^*$ or $(*)$ also holds for some recursive call.
This implies that some recursive call will eventually produce $S^*$.

So assume now $(*)$ holds with $S \neq S^*$. Since $S$ is then not a temporal $(s,z)$-separator,
\textproc{GetSeparator}($S$, $i$) iterates through all scanline separators in~$\seplist[i]$ through $\seplist[\tau]$,
one of which must be the first scanline separator~$S^*_t$ which is not already contained in~$S$.
When it gets to that separator it makes a recursive call \textproc{GetSeparator}($S \cup S^*_t$, $i+1$).
This recursive call then again satisfies $(*)$.
\qed
\end{proof}
We leave open whether \TS{} is \fpt{} or becomes W[1]-hard when 
parameterized by either only the
separator size~$k$ or only the sum~$d_\Sigma$ of Kendall tau distances of permutations of consecutive layers.

\section{Conclusion}\label{chap:conclusion}

We showed that \TS{} remains NP-complete on temporal split graphs even when there are only
$\tau\geq 4$~layers,
but it becomes \fpt{} when parameterized by the lifetime~$\tau$ 
combined with the number~$p$ of ``switching vertices'', that
is, vertices that switch between the independent set and the clique.
We leave open, however, 
whether one can obtain fixed-parameter tractability when only
parameterizing by~$p$. 
Another natural restriction we 
can place on temporal split graphs is limiting the size
of the independent set for all layers.
We conjecture that \TS{} is \fpt{} with respect to the maximum size of the
independent set.

We also showed that \TS{} remains NP-complete on temporal permutation graphs, 
but becomes \fpt{} with
respect to the separator size~$k$ plus the sum~$d_\Sigma$ of Kendall tau distances of permutations of consecutive layers. 
We left open the complexity of \TS{} on temporal
permutation graphs when parameterized by the lifetime~$\tau$. 
Whether the problem stays \fpt{} or becomes W[1]-hard when 
parameterized by either only the
separator size~$k$ or only the sum~$d_\Sigma$ of Kendall tau distances of permutations of consecutive layers remains 
open as well.

Lastly, we leave for future research whether our results also hold in the
\emph{strict} case, that is, when the temporal paths that are to be destroyed by
the separator have strictly increasing time labels. Most of our algorithms
heavily rely on the fact that temporal paths may use several time edges with the
same label, and hence they can presumably not be adapted to the strict setting
in a straightforward way.

\bibliographystyle{abbrvnat}
\bibliography{strings-long,bibliography}

%\clearpage
%\appendix
%\section{Additional Proofs}
%\appendixProofs

\end{document}